\documentclass[12pt]{iopart}
\usepackage[utf8]{inputenc}
\usepackage{dsfont}
\usepackage{xcolor}
\usepackage{amssymb, amsthm}
\usepackage{graphicx}
\newtheorem{theorem}{Theorem}
\begin{document}
\newcommand{\eqref}[1]{(\ref{#1})}
\title{The First Exit Time Statistics and the Entropic Forces in Single File Diffusion}

\author{Alessio Lapolla}

\address{Santa Fe Institute, 1399 Hyde Park Road, Santa Fe, NM, 87501 USA}
\ead{alessiolapolla@santafe.edu}
\vspace{10pt}
\begin{indented}
\item[]
\end{indented}

\begin{abstract}
Single file systems are simplified models to study effectively one-dimensional physical systems. Here we compute analytically the complete first exit time statistics for an ideal overdamped single file with absorbing boundary conditions. Then we use these results to study the speed-accuracy trade-off characterizing this observable in terms of the means square displacement and the entropic repulsive forces of the system.
\end{abstract}

\section{Introduction}
The dynamics of many systems is de facto uni-dimensional. Gene regulation along DNA~\cite{li_effects_2009, ahlberg_many-body_2015}, transport in zeolites~\cite{chou_entropy-driven_1999}, biological channel~\cite{hummer_water_2001}, superionic conductors~\cite{richards_theory_1977}, and the dynamics of colloids~\cite{herrera-velarde_superparamagnetic_2008, lutz_single-file_2004, taloni_single_2017, locatelli_single-file_2016} are examples in which the dimensional constraints on the motion of the particles, and the subsequent impossibility or difficulty in overtaking each other, are well described by single file models.
\begin{figure}[h!]
    \centering
    \includegraphics[width=0.9\textwidth]{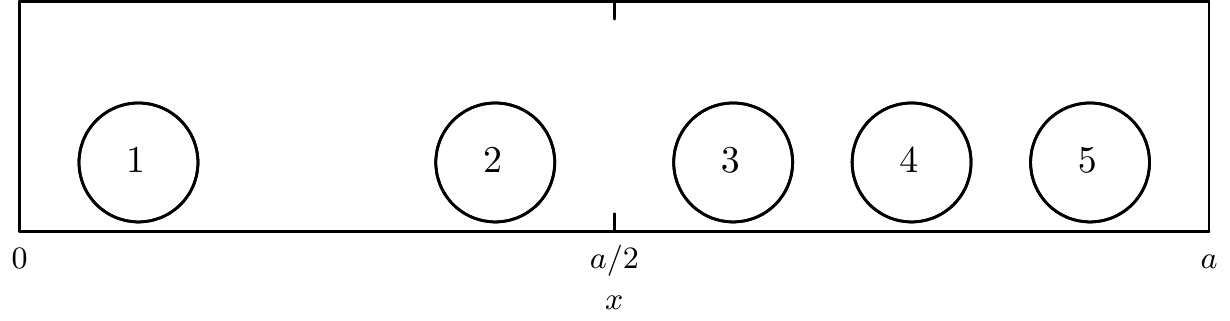}
    \caption{Schematic of a diffusing single file. The left boundary can be both absorbing or reflective, while the right one is always absorbing. The particles are enumerated from left to right. The diffusing process ends when a particle hits an absorbing boundary.}
    \label{fig:Picture}
\end{figure}
These theoretical models describe the diffusing one-dimensional dynamics of particles (or rods) subject to hard-core interactions (though other repulsive potentials have been analyzed in the literature as well~\cite{poncet_generalized_2021}) that force the particles to be ordered at all times. The first successful tentatives to approach the problem analytically have been done more that fifty years ago by Jepsen~\cite{jepsen_dynamics_1965} and Harris~\cite{harris_diffusion_1965} for very large or infinite systems. Remarkably Harris found that the dynamics of a tagged particle in these systems is subdiffusive, i.e. the mean square displacement (MSD) of a tagged particle scales as $\sim\sqrt{t}$ and not linearly as prescribed by the Einstein relation. In the subsequent decades a big effort has been put in place to study various aspects of the system: under the lens of anomalous diffusion~\cite{barkai_theory_2009, barkai_diffusion_2010,fouad_anomalous_2017, leibovich_everlasting_2013, kollmann_single-file_2003}, that as been confirmed for single file system of finite size as well~\cite{lizana_single-file_2008, lizana_diffusion_2009}, its large deviations properties~\cite{krapivsky_large_2014, krapivsky_dynamical_2015}, its connection with the fractional Brownian motion~\cite{lizana_foundation_2010}, its ageing dynamics~\cite{metzler_ageing_2014}, its memory properties~\cite{lapolla_unfolding_2018}, and also the setups in which the tracked particle is driven or active have been studied~\cite{locatelli_active_2015, ryabov_single-file_2011}. In this article we focus on single file models with a limited number of particles subject to at least one absorbing boundary diffusing in a finite interval. First passage time problems in single files have been the object of study for both infinite\cite{sanders_first_2012, ryabov_single-file_2011, ryabov_survival_2012} and finite systems~\cite{rodenbeck_calculating_1998, ryabov_single-file_2013, locatelli_single-file_2016}. In the study of many-body systems with absorbing boundary conditions the hitting time related to the absorption of a particle can be defined in several ways: i) the hitting time can be determined by a specific tagged particle assuming that other particles can be absorbed~\cite{ryabov_single-file_2013}, ii) it can be given by the time at which the last particle is absorbed~\cite{locatelli_single-file_2016}, iii) it can be designated by the time in which a tagged particle is absorbed assuming that the other particles are not affected by the absorbing boundary condition~\cite{kantor_anomalous_2007}, or iv) it can be set by the first particle hitting the boundary~\cite{grebenkov_single-particle_2020}. In this article we will concentrate on the last case, that is: when the earliest particle is absorbed the entire process stops. After the mathematical outline of the model, we will obtain analytically the first exit time probability density and its moments assuming constrained initial conditions. We will use these results to make a connection between the variance of the first exit time and the MSD of the tagged particle, and to comment on how these observables are influenced by the entropic forces characterizing the system.

\section{The single file model and the first exit statistics}
The overdamped dynamics of a single file of $N$ elements with diffusion coefficient $D$ in a interval of length $a$ (see Fig.~\ref{fig:Picture} for a schematic representation) can be described using the following Fokker-Plank equation:
\begin{equation}
    \left\{\partial_t-D\sum_{i=1}^N\partial_{x_i}^2\right\}G(\mathbf{x},t|\mathbf{x}_0)=\delta(\mathbf{x}-\mathbf{x}_0),
  \label{Markovian-FPE}
 \end{equation}
 where the Dirac's delta specifies the initial position of the $N$ particles. The non-crossing conditions are defined by the zero-flux equations 
 \begin{equation}
    \left(\partial_{x_{i+1}}-\partial_{x_i}\right)G(\mathbf{x},t|\mathbf{x}_0)\left|_{x_{i+1}=x_i}\right.=0
  \quad \forall t;\\
 \end{equation}
 while the $N$th particle must satisfy the absorbing boundary condition
 \begin{equation}
     G(\mathbf{x},t|\mathbf{x}_0)\left|_{x_N=a}\right.=0,
 \end{equation}
 the first particle can satisfy either a reflective
 \begin{equation}
     \partial_{x_1}G(\mathbf{x},t|\mathbf{x}_0)\left|_{x_1=0}\right.=0
 \end{equation}
 or an absorbing boundary condition
  \begin{equation}
     G(\mathbf{x},t|\mathbf{x}_0)\left|_{x_1=0}\right.=0.
 \end{equation}
 In order to lighten the notation we consider point particles, the case of finite size rods is equivalent up to a linear rescaling~\cite{lizana_single-file_2008}.\\
 
 The solution of these many-body problems can be written in the form of a series expansion~\cite{gardiner_c.w._handbook_1985, lapolla_bethesf_2020} involving the eigenfunctions\footnote{The two problems considered in this article are hermitian, therefore the left and right eigenfunctions are identical.} $\Psi_\mathbf{k}(\mathbf{x})$ and the eigenvalues $\Lambda_\mathbf{k}$. These elements can be expressed in terms of the eigenfunctions $\psi_{k}(x)$ and the eigenvalues $\lambda_{k}$ solving the respective single particle problems (see~\ref{single particle}). Such that the Green's function (or propagator) that solves Eq.~\eqref{Markovian-FPE} reads:
 \begin{eqnarray}
   &G(\mathbf{x},t|\mathbf{x}_0)=\sum_{\mathbf{k}} \Psi_\mathbf{k}(\mathbf{x})\Psi_\mathbf{k}(\mathbf{x}_0)\mathrm{e}^{-\Lambda_\mathbf{k}t},\\
   &\Psi_\mathbf{k}(\mathbf{x})=\hat{O}_\mathbf{x}\sum_{\mathbf{\pi_k}} \prod_{i=1}^N \psi_{k_i}(x_i),\\
   &\Lambda_\mathbf{k}=\sum_{i=1}^N \lambda_{k_i},
 \end{eqnarray}
 where the zero-flux conditions are taken into account via the ordering operator $\hat{O}_\mathbf{x}$, ensuring, at every time, that $x_1\leq x_2\leq\cdots\leq x_N$; while $\sum_{\mathbf{\pi_k}}$ denotes the sum over all the permutations of the multiset $\mathbf{k}$ of the single particle's eigenvalues $k_i$.\\
The survival function for the entire system, i.e. the probability that there are $N$ particles in the interval $\Arrowvert0,a\arrowvert$ or $\arrowvert0,a\arrowvert$\footnote{Where the simbols $\Arrowvert$ and $\arrowvert$ denote a reflecting or absorbing boundary condition respectively.} beyond time $t$, is given by~\cite{redner_guide_2007}:
\begin{equation}
    S(t,\mathbf{x}_0)=\int_0^a \rmd \mathbf{x} G(\mathbf{x},t|\mathbf{x}_0)=\sum_{\mathbf{k}} \Phi_\mathbf{k}\Psi_\mathbf{k}(\mathbf{x}_0)\mathrm{e}^{-\Lambda_\mathbf{k}t},
    \label{survival function}
\end{equation}
where the integral $\int_0^a \rmd \mathbf{x}$ is over the hyperconic configuration space of the single file problem.
The $\Phi_\mathbf{k}$ factors can be easily computed using the single particle eigenfunctions 
\begin{equation}
    \Phi_\mathbf{k}=N! \prod_{i=1}^N \int_0^a \rmd x \psi_{k_i}(x),
    \label{survival weights}
\end{equation}
where the factorial exploits the exchange symmetry of the system and avoids the costly sum over all the permutations of $\mathbf{k}$, that is to say we can treat each particle as independent and then account for all the possible equivalent rearrangements. The first exit time probability can be readily computed considering~\cite{redner_guide_2007}
\begin{equation}
    F(t,\mathbf{x}_0)=-\frac{\partial S(t,\mathbf{x}_0)}{\partial t}.
    \label{def first exit}
\end{equation}

The previously obtained results require the ability of control the initial positions of all the particles. This could be both very challenging to test experimentally and hard to interpret theoretically. In this paper we will focus on a single parameter initial condition. The $\mathcal{T}$th particle is constrained in a specific initial position $x_{\mathcal{T} 0}$, while the initial positions of the $N_L$ particles to the left and the $N_R$ particles to the right are uniformly distributed between $(0,x_{\mathcal{T} 0})$ and $(x_{\mathcal{T} 0},a)$ respectively.
In this setting the survival function can be obtained computing the marginal of Eq.~\eqref{survival function}:
\begin{equation}
    S(t,x_{\mathcal{T}0})=\int d\mathbf{z} \delta(z_\mathcal{T}-x_{\mathcal{T} 0}) S(t,\mathbf{z}).
\end{equation}
In order to solve the previous integral we must compute the following integrals:
\begin{equation}
  V_{\mathbf{k}}(x_{\mathcal{T} 0})=\int d\mathbf{z} \delta(z_\mathcal{T}-x_{\mathcal{T} 0}) \Psi_\mathbf{k}(\mathbf{z}).
  \label{marginalization}
 \end{equation}
The computation of the $V_{\mathbf{k}}$ terms is in principle expensive due to the presence of all the permutations of $\mathbf{k}$, however using the "extended phase space integration" method~\cite{lizana_diffusion_2009}, it is possible to solve these integrals using formulae involving uniquely the single particle eigenfunctions~\cite{lapolla_unfolding_2018, lapolla_bethesf_2020}, thus the solution to the previous integrals are:
\begin{equation}
   V_{\mathbf{k}}(x_{\mathcal{T} 0})=\frac{N!}{(\mathcal{T}-1)!(N-\mathcal{T})!}\sum_{\{\mathbf{k}\}}\psi_{k_\mathcal{T}}(x_{\mathcal{T} 0})\prod_{j=1}^{\mathcal{T}-1}L_j(x_{\mathcal{T} 0})\prod_{j=\mathcal{T}+1}^N R_j(x_{\mathcal{T} 0}),
   \label{weights}
 \end{equation}
 where
\begin{eqnarray}
    &L(x)=\int_0^x dz \psi_k(z),\\
    &R(x)=\int_{x}^a dz \psi_k(z).
\end{eqnarray}
Therefore the survival and first exit time functions read:
\begin{eqnarray}
    &S(t,x_{\mathcal{T}0})=P_0(x_{\mathcal{T}0})^{-1}\sum_{\mathbf{k}} \Phi_\mathbf{k}V_\mathbf{k}(x_{\mathcal{T}0})\mathrm{e}^{-\Lambda_\mathbf{k}t},\\
    &F(t,x_{\mathcal{T}0})=P_0(x_{\mathcal{T}0})^{-1}\sum_{\mathbf{k}} \Lambda_\mathbf{k} \Phi_\mathbf{k}V_\mathbf{k}(x_{\mathcal{T}0})\mathrm{e}^{-\Lambda_\mathbf{k}t}.
\end{eqnarray}
where 
\begin{equation}
    P_0(x)=\frac{N!}{(N-\mathcal{T})!(\mathcal{T}-1)!a^N}x^{(\mathcal{T}-1)}(a-x)^{N-\mathcal{T}},
\end{equation}
is the normalization that correctly takes into account the initial condition~\cite{lizana_diffusion_2009}.\\
Clearly all expectation moments for any first exit time probability density function are easily computable via direct integration:
\begin{equation}
    \langle t^n\rangle_{x_{\mathcal{T}0}} =\int_0^\infty \rmd t t^n F(t,x_{\mathcal{T}0})=n!P_0(x_{\mathcal{T}0})^{-1}\sum_{\mathbf{k}} \Lambda_\mathbf{k}^{-n} \Phi_\mathbf{k}V_\mathbf{k}(x_{\mathcal{T}0})\mathrm{e}^{-\Lambda_\mathbf{k}t}.
\end{equation}
 
 \section{First exit time analysis}
 It has been shown that often the mere study of the mean first passage time is insufficient to give the entire picture of the process~\cite{godec_first_2016, mattos_first_2012}, thus the study of higher moments and/or the entire probability distribution can be necessary in order to understand the physical system at hand. This has been shown to be particularly important when the dependence on the initial condition is non trivial~\cite{mattos_first_2012}.
 \begin{figure}
     \centering
     \includegraphics{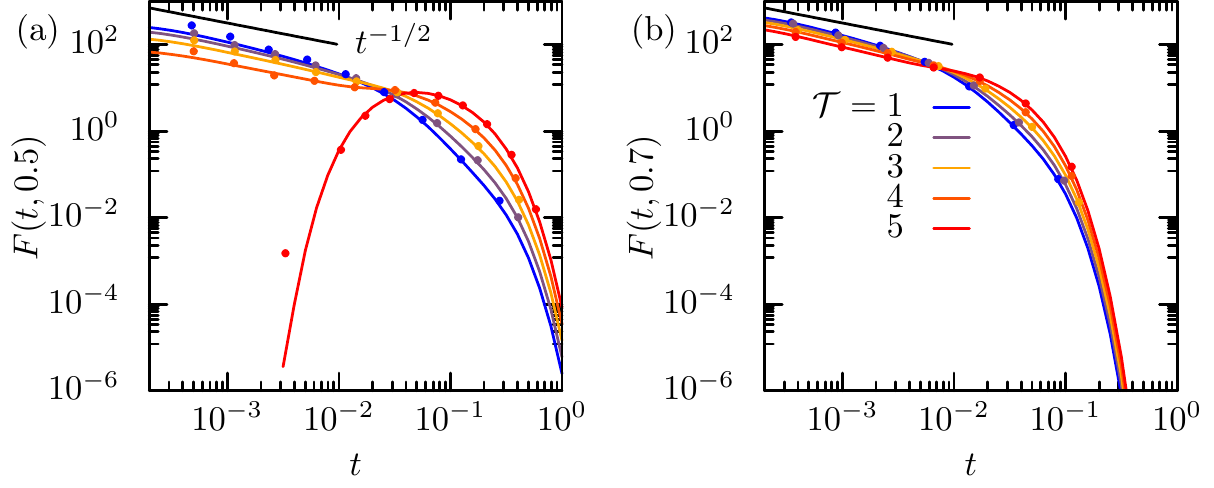}
     \caption{In this figure we show the first exit time probability density for a single file of $5$ particles diffusing in a interval of unitary length and unitary diffusion coefficient. In panel (a) we consider a single absorbing boundary condition while in panel (b) both ends of the interval are absorbing. The numbers labeling the different lines indicate which particle is constrained at $t=0$. For short times the power law decay is clearly visible in both panels. The dots in these panels are the results obtained simulating the process for $500000$ trajectories with a time-step of $10^{-5}$, the simulation algorithm can be found in~\cite{lapolla_bethesf_2020, lapolla_httpsgitlabcomsantafe1singlefileabsorbing_2022}.}
     \label{fig:FPT}
 \end{figure}
 The first exit time density for constrained initial conditions is showed in Fig.~\ref{fig:FPT}. For a single absorbing boundary conditions, the case in which the $N$th particle is constrained initially is qualitatively different from the case in which any of the other particles is. In the former situation the function has the "classic" shape of a first passage time probability density function: it is negligibly small at short times, then increases rapidly to a maximum, to then decay exponentially to zero as the time goes by. The short time behavior is a reflection of the finite time necessary for the last particle to cover the distance between the initial position and the absorbing boundary, this part describes the few trajectories that basically move directly towards the cliff. On the other hand the long time behaviour is given by the few very long trajectories that take a very long time to escape. Conversely, the case in which any of the other particles is constrained initially shows a power-law decay ($\sim t^{-1/2}$) at very short times. This initial trend is due to the fact that the $N$th particle initial position can be arbitrarily close to the absorbing edge, and so the stopping time of the process is analogous to the one of a simple Brownian particle without any interaction. In fact, in Fig.~\ref{fig:FPT}(a) is visible how in the case in which a particle more to the left (the first for example) is constrained initially the process is more likely to finish earlier than in the case in which a particle more to the right (e.g. the fourth) is locked initially. The reason why is that the last particle starts, on average, closer to the boundary since more particles must be accommodated in the interval $(x_0,a)$ initially. This can be shown explicitly by looking at the mean exit time as a function of the initial position of the constrained particle (Fig.~\ref{fig:Var}(a)) that increases monotonically from the leftmost to the rightmost particle (and obviously decreases monotonically as the initial position approaches the cliff).
 Conversely, if both boundaries are absorbing (Fig.~\ref{fig:FPT}(b)) the shape of $F(t,x_{\mathcal{T}0}$) does not present any special case. Even when the first or the last particle is constrained, the last or the first can start arbitrarily close to the other boundary\footnote{Note that in this case it is not possible to know which particle reaches the boundaries the earliest, since we are not computing any conditional first exit time probability.}.\\
 A more interesting phenomenology can be found comparing the mean and the variance of the processes we are examining.
 \begin{figure}
     \centering
     \includegraphics{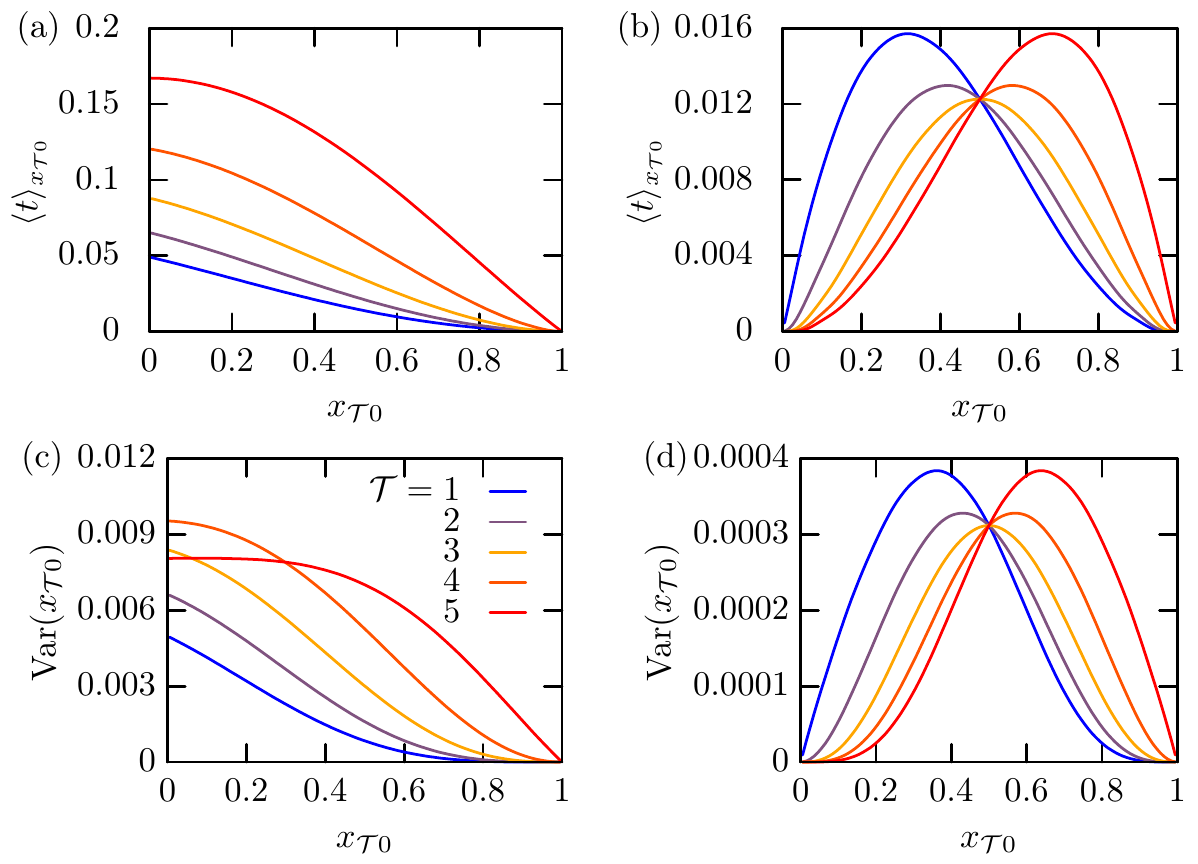}
     \caption{In the panels (a) and (c)  we show the first exit time mean and variance for a single file of $5$ elements diffusing in unitary interval and $D=1$ assuming one absorbing boundary, while in panels (b) and (d) we have both boundaries absorbing. Each number in the legend labels the particle that is constrained at $t=0$.}
     \label{fig:Var}
 \end{figure}
 When both ends of the diffusing interval are absorbing, the variance in Fig.~\ref{fig:Var}(d) matches the behavior of the mean in Fig.~\ref{fig:Var}(b). Both plots are specular around the middle point of the diffusing interval. Intuitively when the first particle starts closer to $0$ it takes on average more time to the system to lose one particle compared to the case in which the last particle is constrained in the same position, since the first particle initial position is very close to a boundary. The situation is reversed on the opposite side of the interval.
 Instead, the results on the interval $\Arrowvert 0,a\arrowvert$ presented in Fig.~\ref{fig:Var}(a) and (c) are less intuitive. The red line, indicating the case in which the last particle is constrained initially, crosses the lines referring to the situations in which the fourth and third particle are constrained initially in Fig.~\ref{fig:Var}(c), while it does not in Fig~\ref{fig:Var}(a). In other words when the $5$th particle is constrained, although taking on average more time to complete, is more accurate then the latter cases, i.e.: the distribution of the exit time is narrower and the stopping time of different realizations less heterogeneous. This crossover is a general feature of the single file with a reflecting and one absorbing boundary condition.
  \begin{figure}
     \centering
     \includegraphics{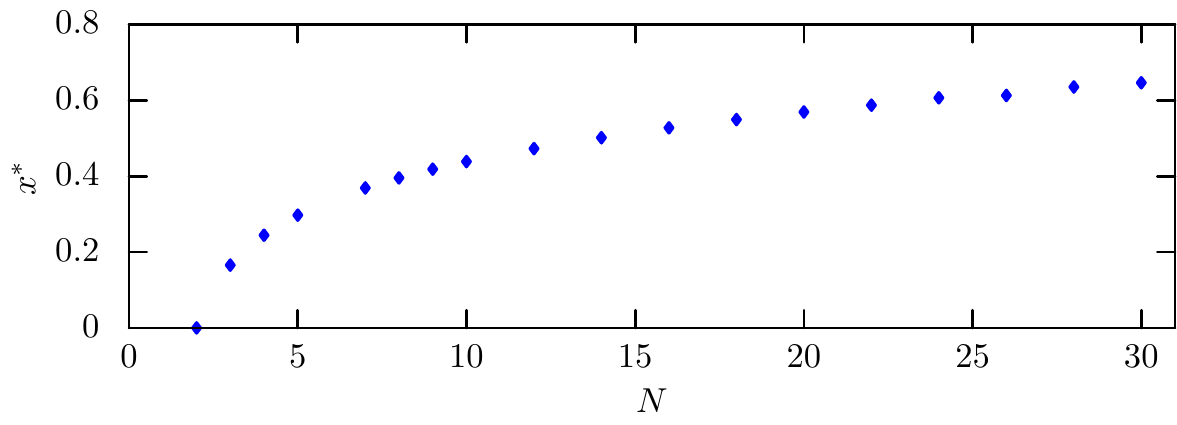}
     \caption{The intersection point $x^*$ (found using the bisection method) between the variance constraining the $N$th and the $N-1$th particle for the case with one single absorbing boundary condition for different sizes of the single file.}
     \label{fig:Intersection}
 \end{figure}
 In Fig.~\ref{fig:Intersection} we plot the intersection points $x^*$ between the variances of the instances in which the last and the second last particles are constrained for different sizes of the single files. The two particles case is the limiting one since the crossover only happens at $x=0$, on the other hand as the number of particle increases the intersection point increases rapidly initially to then slow down as the density increases. Unfortunately, for really large $N$, we cannot conclude if the limiting value is the size of the interval or an intermediate value due to computational constraints.
 
\section{The mean square displacement and the entropic forces}
In order to investigate the reason behind the crossover shown in Fig.~\ref{fig:Var}, we move from a "static" observable, like the first exit time statistic, to a "dynamical" one: the mean square displacement (MSD). The MSD is a "dynamical" observable, in the sense that the trajectory of the tagged particle must be followed over time; the first exit time is not since only the absorbing point(s) must be observed in order to obtain the necessary statistic. In addition, the MSD is a prototypical non-Markovian quantity that refers to the tagged particle and not to the entire systems~\cite{lapolla_manifestations_2019}. Therefore the exchange symmetry of the system cannot be exploited fully to simplify and speed up the computation, as in Eq.\eqref{survival weights} and can be used only partially~\cite{lapolla_bethesf_2020}.\\
The MSD is a common measure to calculate how much a particle deviates from its initial position on average. It is possible to analytically compute the MSD for a given tagged particle until the earliest particle is absorbed. Via the marginalization procedure~\eqref{marginalization}, the Green's function for the tagged particle reads~\cite{lapolla_bethesf_2020}:
\begin{equation}
    \mathcal{G}(x_{\mathcal{T}},t|x_{\mathcal{T}0})=P_0^{-1}(x_{\mathcal{T}0})\sum_\mathbf{k}V_\mathbf{k}(x_\mathcal{T})V_\mathbf{k}(x_{\mathcal{T}0})\mathrm{e}^{-\Lambda_\mathbf{k}t}.
\end{equation}
The MSD can be calculated solving the integral
\begin{equation}
    \langle(x_{\mathcal{T}}-x_{\mathcal{T}0})^2\rangle=\int_0^a \rmd x_{\mathcal{T}} (x_{\mathcal{T}}-x_{\mathcal{T}0})^2 \mathcal{G}(x_{\mathcal{T}},t|x_{\mathcal{T}0}).
\end{equation}
 Clearly the challenging part of the above equation is given by the following integrals:
 \begin{equation}
     \int_0^a \rmd x x^n V_\mathbf{k}(x).
     \label{msd integral}
 \end{equation}
 Their solution is lengthy, though it just involves elementary functions, and is presented in~\ref{MSD integrals}.\\
\begin{figure}
    \centering
    \includegraphics{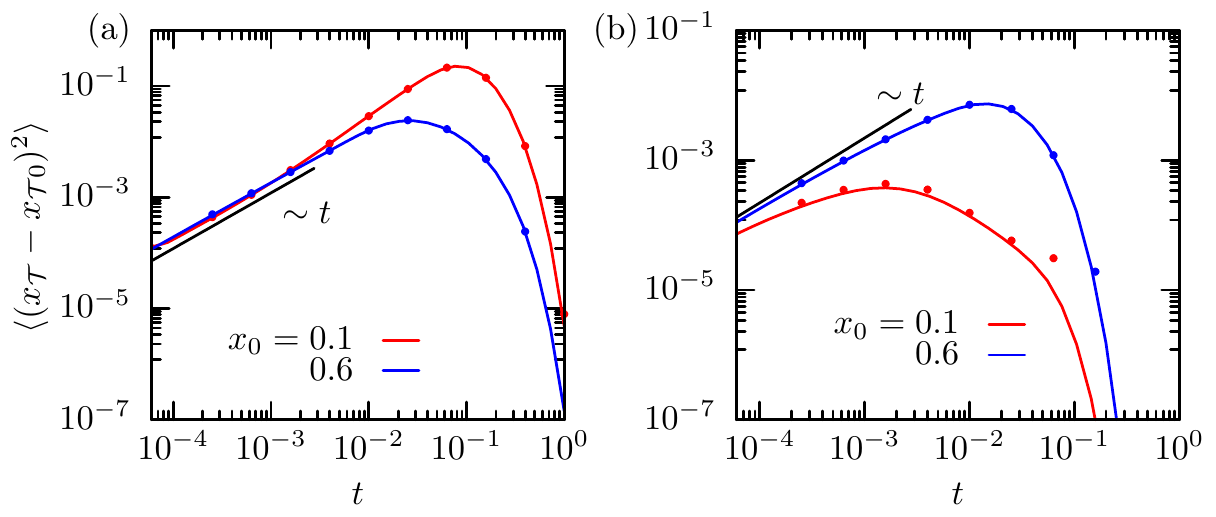}
    \caption{In the left panel we show the MSD for the last particle in a single file with $N=5$ and a single absorbing boundary condition for different initial conditions. The red lines represents the case in which the initial condition is to the left of $x^*\simeq0.3$ while the blue line the opposite case. The red line shows a faster-than-linear growth for intermediate times before the onset of the exponential decay due to the absorption events, while the increase is only linear for the blue line. This is the hallmark of the strong contribution of the entropic forces to the dynamics of the system in the former configuration. In the right panel we show the same quantity for the case in which both boundaries are absorbing. In this panel the blue line only shows a linear growth while in  the case of the red line the first particle of the single file the earliest absorption event happens so rapidly the MSD decays almost immediately. The circles show results obtained using Brownian dynamics simulations.}
    \label{fig:MSD}
\end{figure}
In Fig.~\ref{fig:MSD}(a) we focus on the initial condition dependence of the MSD of the last particle, if we consider only a single cliff this particle is surely responsible of the end of the process. When the initial condition is to the left of the crossing point $x^*$ (Fig.~\ref{fig:Var}(c)) all the particle are initially crammed in a small interval between the reflecting wall and the position of the last particle. After a very short transient in which the tagged particle does not experience any clash on average and behaves diffusively~\cite{lizana_single-file_2008}, the unfavourable entropic configuration influences the collision between particles in such a way that the last particle experiences entropic forces that restrict the number of available configurations of the system. In other words the last particle is more likely to move to the right because of the hard-core reactions. Nevertheless the dynamics we are analyzing is overdamped, therefore the forces due to momentum transfers between particles are absent, and the effect of the entropic forces does not extend to the mean first exit time. I.e. the "push" does not move the last particle to the right, it just prevents the particle to diffuse too much to the left. On the other hand if the starting condition is to the right of $x^*$, the effect of these forces in not dominant anymore, the growth of the MSD is slower, and it does not show a noticeable super-linear increase. Oppositely, if both boundary conditions are absorbing, the dramatic effect of the entropic forces completely disappears (see Fig.~\ref{fig:MSD}(b)) since both the first and the last particle can escape and the MSD is much smaller. Moreover the range of this strong effect of the entropic forces is proportional to the number of particles in the interval as it is possible to infer from Fig.~\ref{fig:Intersection}.

\section{Conclusions} 
In this article we presented analytical solutions for the first exit statistics and the MSD for a finite single file. Focusing on this diffusing system without external forces we were able to focus on the effect of the entropic forces deriving uniquely by the hard-core interaction between the particles. The analytical solutions presented here can be straightforwardly expanded to single file models diffusing in a generic external potential as well, as long as the eigendecomposition of the corresponding single particle problem is known. The most striking effect of these entropic forces is the realization of a class of initial configurations in which a trade-off between the speed and the accuracy of the first exit process can be engineered. In other words there are cases in which a configuration has a smaller first exit time with respect to a second one but it has also a larger variance. This trade-off has been shown to be relevant in first passage problem in biological settings~\cite{banerjee_elucidating_2017} and here we focused on the dependence of this phenomenon on the initial condition.  In our model this is characterized by a qualitative change in the behaviour of the MSD, in particular, the more accurate process is marked by a faster-than-linear growth of the MSD for intermediate times.
Beyond the mere mathematical techniques to obtain the observable described in this article, we think that this work could help shine new light on the trade-off between speed and accuracy in first passage time problems.

\section*{Acknowledgments}
 The author is grateful to Sidney Redner for the invaluable suggestions. The author  gratefully acknowledges financial support from NSF grant DMR-1910736.
 \appendix
 \section{The single particle problems}
  \label{single particle}
 The eigenexpansion method~\cite{gardiner_c.w._handbook_1985} is a common and powerful technique to solve Fokker-Plank equations, in our case the single particle diffusion equation reads:
 \begin{equation}
     \frac{\partial G(x,t|x_0)}{\partial t}=D\frac{\partial^2 G(x,t|x_0)}{\partial x^2},
 \end{equation}
 with localized initial conditions $G(x,0|x_0)=\delta(x-x_0)$. The solution to such equation is given by the series
 \begin{equation}
     G(x,t|x_0)=\sum_{k}\psi_k(x)\psi_k(x_0)\mathrm{e}^{-\lambda_k t},
 \end{equation}
 where the eigenfunctions and the eigenvalues are defined with the help of the boundary conditions.\\
 In the case in which both boundaries are absorbing, that is
 \begin{equation}
     G(0,t|x_0)=G(a,t|x_0)=0,
 \end{equation}
 we have that
 \begin{eqnarray}
     &\lambda_k=D\pi^2 (k+1)^2/a^2,\\
     &\psi_k(x)=\sqrt{2/a}\sin((k+1)\pi x/a).
 \end{eqnarray}
 Via direct integration it is then possible to obtain:
 \begin{eqnarray}
     &L_k(x)=\sqrt{2a}\frac{1-\cos((k+1)\pi x/a)}{(k+1)\pi},\\
     &R_k(x)=\sqrt{2a}\frac{(-1)^k-\cos((k+1)\pi x/a)}{(k+1)\pi},\\
     &\Phi_\mathbf{k}=N!\prod_{i=1}^N \sqrt{2a}\frac{1-(-1)^{k_i+1}}{(k_i+1)\pi}.
 \end{eqnarray}
 While in the case of a single absorbing boundary condition
 \begin{equation}
     \frac{\partial G(x,t|x_0)}{\partial x}\left|_0\right.=G(a,t|x_0)=0,
 \end{equation}
 the eigenfunctions and eigenvalues are
 \begin{eqnarray}
     &\lambda_k=\frac{D\pi^2(2k+1)^2}{4a^2},\\
     &\psi_k(x)=\sqrt{2/a}\cos(\frac{(2k+1)\pi x}{2a}),
 \end{eqnarray}
 and
 \begin{eqnarray}
    &L_k(x)=2\sqrt{2a}\frac{\sin((2k+1)\pi x/(2a))}{(2k+1)\pi},\\
    &R_k(x)=2\sqrt{2a}\frac{(-1)^k-\sin((2k+1)\pi x/(2a))}{(2k+1)\pi},\\
    &\Phi_\mathbf{k}=N!\prod_{i=1}^N 2\sqrt{2a}\frac{\sin((k_i+1)\pi/a)}{(k_i+1)\pi}.
 \end{eqnarray}

 \section{MSD integrals}
 \label{MSD integrals}
 The integrals presented in Eq.~\eqref{msd integral} are in principle hard to solve for a generic size $N$ of the single file. Keeping in mind Eq.~\eqref{weights}, it is clear that the integrand of Eq.~\eqref{msd integral} is composed by a monomial, two different products of the $L_k(x)$ and $R_k(x)$ functions, and the single particle eigenfunctions. In the two cases presented in this article all these functions are elementary trigonometric functions for which the generalized product-to-sum formulae (see~\ref{product-to-sum}) can be proven. These formulae can be use to convert the single integral involving a complicated finite product in Eq.~\eqref{msd integral} to a finite sum over $M$ terms of a monomial and up to three elementary trigonometric functions $t_k(x)$:
 \begin{equation}
     \int_0^a \rmd x x^n V_\mathbf{k}(x)=\sum_{j=1}^M w_j \int_0^a \rmd x x^n t_{\alpha_j}(x)t_{\beta_j}(x)t_{\gamma_j}(x),
     \label{msd expansion}
 \end{equation}
 where the weights $w_j$ take into account the prefactor for each term. The calculation of these prefactors is complicated by the fact that $L_k(x)$ and $R_k(x)$ (in the double absorbing case) and $R_k(x)$ (in the single adsorbing case) are composed by two addends (see~\ref{single particle}); however using the binomial theorem the product over these functions can be reduced to the case presented in Eq.~\eqref{msd expansion}. In last this case a integrand in Eq.~\eqref{msd expansion} presents only one or two trigonometric factors.\\
 For example, let $a=1$, $\mathbf{k}=\{k_1,k_2,k_3,k_4,k_5\}$, and let us consider the single absorbing case tagging the 3rd particle. Then the integral in Eq.~\eqref{msd expansion} reads:
 \begin{eqnarray}
    &\int_0^1 \rmd x x^n V_\mathbf{k}(x)=\nonumber\\
    &\int_0^1 \rmd x x^n \sqrt{2}\cos(\frac{(2k_3+1)\pi x}{2})\prod_{j\in\{k1,k2\}}\frac{2\sqrt{2}\sin(\frac{(2j+1)\pi x}{2})}{\pi(2j+1)}\times\nonumber\\
    &\prod_{j\in\{k4,k5\}}\frac{2\sqrt{2}((-1)^j-\sin(\frac{(2j+1)\pi x}{2}))}{\pi(2j+1)}.
    \label{with products}
 \end{eqnarray}
 The first product can be easily simplified using the product-to-sum rule and reads:
 \begin{eqnarray}
     &\prod_{j\in\{k1,k2\}}\frac{2\sqrt{2}\sin(\frac{(2j+1)\pi x}{2})}{\pi(2j+1)}=\frac{4}{\pi^2(1+2k_1)(1+2k_2)}\times\nonumber\\
     &\left(\cos(\frac{[(1+2k_1)-(1+2k_2)]\pi x}{2})-\cos(\frac{[(1+2k_1)+(1+2k_2)]\pi x}{2})\right).
     \label{expanded left}
 \end{eqnarray}
 While the second product is a bit more complicated since we first need to expand it using the binomial theorem:
 \begin{eqnarray}
     &\prod_{j\in\{k4,k5\}}\frac{2\sqrt{2}((-1)^j-\sin(\frac{(2j+1)\pi x}{2}))}{\pi(2j+1)}=\frac{(-1)^{k_4+k_5}8}{\pi^2(1+2k_4)(1+2k_5)}-\nonumber\\
     &\frac{(-1)^{k_5}8\sin(\frac{(2k_4+1)\pi x}{2})}{\pi^2(1+2k_4)(1+2k_5)}-
     \frac{(-1)^{k_4}8\sin(\frac{(2k_5+1)\pi x}{2})}{\pi^2(1+2k_4)(1+2k_5)}+\nonumber\\
     &\frac{8\sin(\frac{(2k_4+1)\pi x}{2})\sin(\frac{(2k_5+1)\pi x}{2})}{\pi^2(1+2k_4)(1+2k_5)};
 \end{eqnarray}
 and then we can apply the product-to-sum rules to each term and obtain:
 \begin{eqnarray}
     &\frac{(-1)^{k_4+k_5}8}{\pi^2(1+2k_4)(1+2k_5)}-\nonumber\\
     &\frac{(-1)^{k_5}8\sin(\frac{(2k_4+1)\pi x}{2})}{\pi^2(1+2k_4)(1+2k_5)}-
     \frac{(-1)^{k_4}8\sin(\frac{(2k_5+1)\pi x}{2})}{\pi^2(1+2k_4)(1+2k_5)}+\nonumber\\
    &\frac{4}{\pi^2(1+2k_4)(1+2k_5)}\times\nonumber\\
    &\left(\cos(\frac{[(1+2k_4)-(1+2k_5)]\pi x}{2})-\cos(\frac{[(1+2k_4)+(1+2k_5)]\pi x}{2})\right).
    \label{expanded right}
 \end{eqnarray}
 Plugging the Eqs.~\eqref{expanded left} and~\eqref{expanded right} in Eq.~\eqref{with products} and performing all the products we can rewrite the integral of the products in a sum of integrals as in Eq.~\eqref{msd expansion}. Note that some addends just present two trigonometric functions in this case, e.g.:
 \begin{equation}
     \int_0^1 \rmd x x^n \sqrt{2}\cos(\frac{(2k_3+1)\pi x}{2})\frac{4\cos(\frac{[(1+2k_1)-(1+2k_2)]\pi x}{2})}{\pi^2(1+2k_1)(1+2k_2)}\frac{(-1)^{k_4+k_5}8}{\pi^2(1+2k_4)(1+2k_5)}.
 \end{equation}
 Thus it is clear that all integrals appearing on the right side of Eq.~\eqref{msd expansion} are easily solvable by parts or with the help of a computer algebra system. Nevertheless they present several sub-cases that must be treated separately (e.g. when $\alpha_j=\beta_j=0$), it is quite involved to keep track of all the possibilities, and the final results are quite lengthy. Consequently we do not report them in the article. The interested reader can find them in the code connected with this publication~\cite{lapolla_httpsgitlabcomsantafe1singlefileabsorbing_2022}.
 Unfortunately the number $M\sim O(2^N)$ (see~\ref{product-to-sum}), henceforth the analytical computation of the MSD with this method is unfeasible for large systems for short times, and a Brownian dynamics simulation is more convenient. Conversely, for small or intermediate system sizes and/or for long times, our analytical method is superior to a computational one since our result is based on a series expansion that unfolds "backward in time", that is few addends of the series are sufficient to describe the long time behavior.
 
 \section{Generalized product-to-sum formulae}
 \label{product-to-sum}
The product-to-sum formulae or Werner formulae 
\begin{eqnarray}
  &\sin(a) \sin(b) = \frac{\cos(a - b) - \cos(a + b)}{2} \label{sin-sin}\\
  &\cos(a) \cos(b) = \frac{\cos(a - b) + \cos(a + b)}{2} \label{cos-cos}\\
  &\sin(a) \cos(b) = \frac{\sin(a + b) + \sin(a - b)}{2} \label{sin-cos}\\
  &\cos(a) \sin(b) = \frac{\sin(a + b) - \sin(a - b)}{2}. \label{cos-sin}
\end{eqnarray}
express the equivalence between the product and a sum of two trigonometric functions. In this section we present analogous formulae for arbitrary long products of sinuses and cosines.These formulae have been surely proven countless times in the history of Mathematics, however, since we were not able to find any reference for the sine formula while the cosine formula can be found in~\cite{abramowitz_milton_and_stegun_irene_a_handbook_1964}, we prove them by induction in this section for sake of completeness.\\
Let $S^N=\{1,-1\}^N$ denote the collection of all possible combinations of binary tuple of length $N$ and $e_i$ the $i$th element of one of these tuples; then we can prove the following:
\begin{theorem}
  The finite product of $N$ cosines can be rewritten as sum according to the following formula:
  \begin{equation}
    \prod_{k=1}^N\cos(\theta_k)=\frac{1}{2^N}\sum_{S^N}\cos(e_1\theta_1+\cdots+e_N\theta_N).
  \end{equation}
\end{theorem}
\begin{proof}
  The theorem can be proven by induction. For $N=1$:
  \begin{equation}
    \cos(\theta_1)=\frac{1}{2}\left(\cos(\theta_1)+\cos(-\theta_1)\right)=\cos(\theta_1).
  \end{equation}
  Then for $N=2$ we can take advantage of the parity properties of the cosine and rewrite Eq.~\eqref{cos-cos} as
  \begin{equation}
    \cos(\theta_1) \cos(\theta_2)=\frac{1}{4}\left(\cos(\theta_1+\theta_2)+\cos(-\theta_1-\theta_2)+\cos(\theta_1-\theta_2)+\cos(-\theta_1+\theta_2)\right).
    \label{parity cosine}
  \end{equation}
  If we assume that the theorem is true for $N$ we can check it for the $N+1$ case and write
\begin{eqnarray}
        &\cos(\theta_{N+1})\prod_{k=1}^N\cos(\theta_k)=\nonumber\\
      &\cos(\theta_{N+1})\frac{1}{2^N}\sum_{S^N}\cos(e_1\theta_1+\cdots+e_N\theta_N)=\nonumber\\
      &\frac{1}{2^N}\sum_{S^N}\frac{1}{4}(\cos(e_1\theta_1+\cdots+e_N\theta_N+\theta_{N+1})+\cos(-e_1\theta_1-\cdots-e_N\theta_N-\theta_{N+1})+\nonumber\\
      &\cos(e_1\theta_1+\cdots+e_N\theta_N-\theta_{N+1})+\cos(-e_1\theta_1-\cdots-e_N\theta_N+\theta_{N+1});
\end{eqnarray}
  where we used Eq.~\eqref{parity cosine} in the second passage. Now we can notice that $\cos(e_1\theta_1+\cdots+e_N\theta_N+\theta_{N+1})=\cos(-e_1\theta_1-\cdots-e_N\theta_N-\theta_{N+1})$ and $\cos(e_1\theta_1+\cdots+e_N\theta_N-\theta_{N+1})=\cos(-e_1\theta_1-\cdots-e_N\theta_N+\theta_{N+1})$ and rearrange the terms in the following fashion
  \begin{eqnarray}
      &\frac{1}{2^N}\sum_{S^N}\frac{1}{4}(\cos(e_1\theta_1+\cdots+e_N\theta_N+\theta_{N+1})+\cos(e_1\theta_1+\cdots+e_N\theta_N-\theta_{N+1}))\nonumber\\
      &+\frac{1}{2^N}\sum_{S^N}\frac{1}{4}(\cos(-e_1\theta_1-\cdots-e_N\theta_N-\theta_{N+1})+\nonumber\\
      &\cos(-e_1\theta_1-\cdots-e_N\theta_N+\theta_{N+1}),
  \end{eqnarray}
  since these two sums are equal we obtain
\begin{eqnarray}
          &\frac{1}{2^N}\sum_{S^N}\frac{1}{2}(\cos(e_1\theta_1+\cdots+e_N\theta_N+\theta_{N+1})+\cos(e_1\theta_1+\cdots+e_N\theta_N-\theta_{N+1}))=\nonumber\\
      &\frac{1}{2^{N+1}}\sum_{S^{N+1}}\cos(e_1\theta_1+\cdots+e_N\theta_N+e_{N+1}\theta_{N+1});
\end{eqnarray}
\end{proof}
Analogously we can prove a similar theorem for the sine function ($\lfloor\cdot\rfloor$ defines the floor function).
\begin{theorem}
  The finite product of $N$ sines can be rewritten as sum according to the following formula:
  \begin{equation}
    \prod_{k=1}^N\sin(\theta_k)=\frac{(-1)^{\lfloor\frac
        {N}{2}\rfloor}}{2^N}\cases{
        \sum_{S^N}\cos(e_1\theta_1+\cdots+e_N\theta_N)\prod_{j=1}^N e_j &if N even\\
        \sum_{S^N}\sin(e_1\theta_1+\cdots+e_N\theta_N)\prod_{j=1}^N e_j &if  N odd.
        }
  \end{equation}
\end{theorem}
\begin{proof}
  If $N=1$ we have
  \begin{equation}
    \sin(\theta_1)=\frac{1}{2}(\sin(\theta_1)-\sin(-\theta_1)),
  \end{equation}
  while if $N=2$
  \begin{eqnarray}
      &\sin(\theta_1)\sin(\theta_2)=\frac{1}{4}(\cos(\theta_1+\theta_2)+\cos(-\theta_1-\theta_2)-\cos(\theta_1-\theta_2)-\cos(-\theta_1+\theta_2))=\nonumber\\
      &\frac{1}{2}(\cos(\theta_1+\theta_2)-\cos(\theta_1-\theta_2))
  \end{eqnarray}
  as prescribed by equation~\eqref{sin-sin}.\\
  Now if $N$ is odd, by induction,
  \begin{eqnarray}
      &\sin(\theta_{N+1})\prod_{k=1}^N \sin(\theta_k)=\nonumber\\
      &\sin(\theta_{N+1})\frac{(-1)^{\lfloor\frac
          {N}{2}\rfloor}}{2^N} \sum_{S^N}\sin(e_1\theta_1+\cdots+e_N\theta_N)\prod_{j=1}^N e_j=\nonumber\\
      &\frac{(-1)^{\lfloor\frac
          {N}{2}\rfloor}}{2^N} \sum_{S^N}\frac{1}{2}(\cos(-e_1\theta_1-\cdots-e_N\theta_N+\theta_{N+1})\prod_{j=1}^N e_j\nonumber\\
      &-\cos(+e_1\theta_1+\cdots+e_N\theta_N+\theta_{N+1})\prod_{j=1}^N e_j)=\\
      &\frac{(-1)^{\lfloor\frac
          {N}{2}\rfloor}(-1)}{2^{N+1}} \sum_{S^N}(-\cos(e_1\theta_1+\cdots+e_N\theta_N-\theta_{N+1})\prod_{j=1}^N e_j\nonumber\\
      &+\cos(e_1\theta_1+\cdots+e_N\theta_N+\theta_{N+1})\prod_{j=1}^N e_j)
  \end{eqnarray}
  where we used Eq.~\eqref{sin-sin} and the parity of the cosine. Now noticing that $(-1)^{\lfloor\frac{N}{2}\rfloor}(-1)=(-1)^{\lfloor\frac{N+1}{2}\rfloor}$ if $N$ is odd we can prove the first part of the theorem.\\
  In a similar fashion if $N$ is even
  \begin{eqnarray}
      &\sin(\theta_{N+1})\prod_{k=1}^N \sin(\theta_k)=\nonumber\\
      &\sin(\theta_{N+1})\frac{(-1)^{\lfloor\frac{N}{2}\rfloor}}{2^N} \sum_{S^N}\cos(e_1\theta_1+\cdots+e_N\theta_N)\prod_{j=1}^N e_j=\nonumber\\
      &\frac{(-1)^{\lfloor\frac{N}{2}\rfloor}}{2^N} \sum_{S^N}\frac{1}{2}(\sin(e_1\theta_1+\cdots+e_N\theta_N+\theta_{N+1})\prod_{j=1}^N e_j+\nonumber\\
      &\sin(-e_1\theta_1-\cdots-e_N\theta_N+\theta_{N+1})\prod_{j=1}^N e_j),
    \end{eqnarray}
    where we used Eq.~\eqref{sin-cos}. Now using the odd property of the sinus and the fact that $\lfloor\frac{N}{2}\rfloor=\lfloor\frac{N+1}{2}\rfloor$ if $N$ is even, we can conclude the proof writing
    \begin{eqnarray}
        &\frac{(-1)^{\lfloor\frac
            {N}{2}\rfloor}}{2^{N+1}} \sum_{S^N}(\sin(e_1\theta_1+\cdots+e_N\theta_N+\theta_{N+1})\prod_{j=1}^N e_j-\nonumber\\
            &\sin(e_1\theta_1+\cdots+e_N\theta_N-\theta_{N+1})\prod_{j=1}^N e_j)=\nonumber\\
        &\frac{(-1)^{\lfloor\frac
            {N+1}{2}\rfloor}}{2^{N+1}} \sum_{S^N}(\sin(e_1\theta_1+\cdots+e_{N+1}\theta_{N+1})\prod_{j=1}^{N+1} e_j.
    \end{eqnarray}
\end{proof}

\section{Unconstrained starting conditions}
An easy realizable initial condition is to do not constrain any particle and just distribute them uniformly in the interval $(0,a)$ at time $t=0$.In this case the survival function is obtained integrating Eq.~\eqref{survival function} over the initial condition. Since now we are integrating over all particle we can fully exploit the exchange symmetry of the system. Hence the survival function reads:
\begin{equation}
    S(t)=\sum_{\mathbf{k}} \Phi_\mathbf{k}\Psi_\mathbf{k}(\mathbf{x}_0)\mathrm{e}^{-\Lambda_\mathbf{k}t}.
\end{equation}
The first exit time distribution can then be obtained via a simple time derivative.\\
Also the MSD can be computed analytically considering the following integral:
 \begin{equation}
    \langle(x_{\mathcal{T}}-x_{\mathcal{T}0})^2\rangle_{\mathrm{UC}}=\int_0^a \rmd x_{\mathcal{T}} \int_0^a \rmd x_{\mathcal{T}0} (x_{\mathcal{T}}-x_{\mathcal{T}0})^2 \mathcal{G}(x_{\mathcal{T}},t|x_{\mathcal{T}0})P_0(x_{\mathcal{T}0}).
\end{equation}
That can be solved using the method described in~\ref{MSD integrals}. The observables for this problem have been implemented in~\cite{lapolla_httpsgitlabcomsantafe1singlefileabsorbing_2022}.

\pagebreak 
\bibliographystyle{unsrt}
\bibliography{main.bib}

\begin{thebibliography}{10}

\bibitem{li_effects_2009}
Gene-Wei Li, Otto~G. Berg, and Johan Elf.
\newblock Effects of macromolecular crowding and {DNA} looping on gene
  regulation kinetics.
\newblock {\em Nat. Phys.}, 5(4):294--297, 2009.

\bibitem{ahlberg_many-body_2015}
Sebastian Ahlberg, Tobias Ambjörnsson, and Ludvig Lizana.
\newblock Many-body effects on tracer particle diffusion with applications for
  single-protein dynamics on {DNA}.
\newblock {\em New J. Phys.}, 17(4):043036, April 2015.

\bibitem{chou_entropy-driven_1999}
Tom Chou and Detlef Lohse.
\newblock Entropy-{Driven} {Pumping} in {Zeolites} and {Biological} {Channels}.
\newblock {\em Phys. Rev. Lett.}, 82(17):3552--3555, April 1999.

\bibitem{hummer_water_2001}
G.~Hummer, J.~C. Rasaiah, and J.~P. Noworyta.
\newblock Water conduction through the hydrophobic channel of a carbon
  nanotube.
\newblock {\em Nature}, 414(6860):188--190, 2001.

\bibitem{richards_theory_1977}
Peter~M. Richards.
\newblock Theory of one-dimensional hopping conductivity and diffusion.
\newblock {\em Phys. Rev. B}, 16(4):1393--1409, August 1977.

\bibitem{herrera-velarde_superparamagnetic_2008}
S.~Herrera-Velarde and R.~Castañeda-Priego.
\newblock Superparamagnetic colloids confined in narrow corrugated substrates.
\newblock {\em Phys. Rev. E}, 77(4):041407, 2008.

\bibitem{lutz_single-file_2004}
Christoph Lutz, Markus Kollmann, and Clemens Bechinger.
\newblock Single-{File} {Diffusion} of {Colloids} in {One}-{Dimensional}
  {Channels}.
\newblock {\em Phys. Rev. Lett.}, 93(2):026001, July 2004.
\newblock Publisher: American Physical Society.

\bibitem{taloni_single_2017}
Alessandro Taloni, Ophir Flomenbom, Ramón Castañeda-Priego, and Fabio
  Marchesoni.
\newblock Single file dynamics in soft materials.
\newblock {\em Soft Matter}, 13(6):1096--1106, 2017.

\bibitem{locatelli_single-file_2016}
Emanuele Locatelli, Matteo Pierno, Fulvio Baldovin, Enzo Orlandini, Yizhou Tan,
  and Stefano Pagliara.
\newblock Single-{File} {Escape} of {Colloidal} {Particles} from {Microfluidic}
  {Channels}.
\newblock {\em Phys. Rev. Lett.}, 117(3):038001, July 2016.

\bibitem{poncet_generalized_2021}
Alexis Poncet, Aurélien Grabsch, Pierre Illien, and Olivier Bénichou.
\newblock Generalized {Correlation} {Profiles} in {Single}-{File} {Systems}.
\newblock {\em Phys. Rev. Lett.}, 127(22):220601, November 2021.

\bibitem{jepsen_dynamics_1965}
D.~W. Jepsen.
\newblock Dynamics of a {Simple} {Many}‐{Body} {System} of {Hard} {Rods}.
\newblock {\em J. Math. Phys.}, 6(3):405--413, 1965.

\bibitem{harris_diffusion_1965}
T.~E. Harris.
\newblock Diffusion with “collisions” between particles.
\newblock {\em Journal of Applied Probability}, 2(2):323--338, December 1965.

\bibitem{barkai_theory_2009}
E.~Barkai and R.~Silbey.
\newblock Theory of {Single} {File} {Diffusion} in a {Force} {Field}.
\newblock {\em Phys. Rev. Lett.}, 102(5):050602, 2009.

\bibitem{barkai_diffusion_2010}
E.~Barkai and R.~Silbey.
\newblock Diffusion of tagged particle in an exclusion process.
\newblock {\em Phys. Rev. E}, 81(4):041129, 2010.

\bibitem{fouad_anomalous_2017}
Ahmed~M. Fouad and Edward~T. Gawlinski.
\newblock Anomalous and nonanomalous behaviors of single-file dynamics.
\newblock {\em Phys. Lett. A}, 381(35):2906--2911, 2017.

\bibitem{leibovich_everlasting_2013}
N.~Leibovich and E.~Barkai.
\newblock Everlasting effect of initial conditions on single-file diffusion.
\newblock {\em Phys. Rev. E}, 88(3):032107, 2013.

\bibitem{kollmann_single-file_2003}
Markus Kollmann.
\newblock Single-file {Diffusion} of {Atomic} and {Colloidal} {Systems}:
  {Asymptotic} {Laws}.
\newblock {\em Phys. Rev. Lett.}, 90(18):180602, May 2003.
\newblock Publisher: American Physical Society.

\bibitem{lizana_single-file_2008}
L.~Lizana and T.~Ambjörnsson.
\newblock Single-{File} {Diffusion} in a {Box}.
\newblock {\em Phys. Rev. Lett.}, 100(20):200601, 2008.

\bibitem{lizana_diffusion_2009}
L.~Lizana and T.~Ambjörnsson.
\newblock Diffusion of finite-sized hard-core interacting particles in a
  one-dimensional box: {Tagged} particle dynamics.
\newblock {\em Phys. Rev. E}, 80(5):051103, 2009.

\bibitem{krapivsky_large_2014}
P.~L. Krapivsky, Kirone Mallick, and Tridib Sadhu.
\newblock Large {Deviations} in {Single}-{File} {Diffusion}.
\newblock {\em Phys. Rev. Lett.}, 113(7):078101, August 2014.
\newblock Publisher: American Physical Society.

\bibitem{krapivsky_dynamical_2015}
P~L Krapivsky, Kirone Mallick, and Tridib Sadhu.
\newblock Dynamical properties of single-file diffusion.
\newblock {\em J. Stat. Mech.}, 2015(9):P09007, September 2015.

\bibitem{lizana_foundation_2010}
Ludvig Lizana, Tobias Ambjörnsson, Alessandro Taloni, Eli Barkai, and
  Michael~A. Lomholt.
\newblock Foundation of fractional {Langevin} equation: {Harmonization} of a
  many-body problem.
\newblock {\em Phys. Rev. E}, 81(5):051118, 2010.

\bibitem{metzler_ageing_2014}
R.~Metzler, L.~Sanders, M.~A. Lomholt, L.~Lizana, K.~Fogelmark, and Tobias
  Ambjörnsson.
\newblock Ageing single file motion.
\newblock {\em Eur. Phys. J. Spec. Top.}, 223(14):3287--3293, 2014.

\bibitem{lapolla_unfolding_2018}
Alessio Lapolla and Aljaž Godec.
\newblock Unfolding tagged particle histories in single-file diffusion: exact
  single- and two-tag local times beyond large deviation theory.
\newblock {\em New J. Phys.}, 20(11):113021, 2018.

\bibitem{locatelli_active_2015}
Emanuele Locatelli, Fulvio Baldovin, Enzo Orlandini, and Matteo Pierno.
\newblock Active {Brownian} particles escaping a channel in single file.
\newblock {\em Phys. Rev. E}, 91(2):022109, February 2015.

\bibitem{ryabov_single-file_2011}
Artem Ryabov and Petr Chvosta.
\newblock Single-file diffusion of externally driven particles.
\newblock {\em Phys. Rev. E}, 83(2):020106, February 2011.

\bibitem{sanders_first_2012}
Lloyd~P. Sanders and Tobias Ambjörnsson.
\newblock First passage times for a tracer particle in single file diffusion
  and fractional {Brownian} motion.
\newblock {\em J. Chem. Phys.}, 136(17):175103, May 2012.

\bibitem{ryabov_survival_2012}
Artem Ryabov and Petr Chvosta.
\newblock Survival of interacting {Brownian} particles in crowded
  one-dimensional environment.
\newblock {\em J. Chem. Phys.}, 136(6):064114, February 2012.

\bibitem{rodenbeck_calculating_1998}
Christian Rödenbeck, Jörg Kärger, and Karsten Hahn.
\newblock Calculating exact propagators in single-file systems via the
  reflection principle.
\newblock {\em Phys. Rev. E}, 57(4):4382--4397, April 1998.

\bibitem{ryabov_single-file_2013}
Artem Ryabov.
\newblock Single-file diffusion in an interval: {First} passage properties.
\newblock {\em J. Chem. Phys.}, 138(15):154104, April 2013.

\bibitem{kantor_anomalous_2007}
Yacov Kantor and Mehran Kardar.
\newblock Anomalous diffusion with absorbing boundary.
\newblock {\em Phys. Rev. E}, 76(6):061121, December 2007.

\bibitem{grebenkov_single-particle_2020}
Denis~S Grebenkov, Ralf Metzler, and Gleb Oshanin.
\newblock From single-particle stochastic kinetics to macroscopic reaction
  rates: fastest first-passage time of {N} random walkers.
\newblock {\em New J. Phys.}, 22(10):103004, October 2020.

\bibitem{gardiner_c.w._handbook_1985}
{Gardiner, C.W.}
\newblock {\em Handbook of {Stochastic} {Methods} for {Physics}, {Chemistry}
  and {Natural} {Sciences}}.
\newblock Springer-Verlag, second edition, 1985.

\bibitem{lapolla_bethesf_2020}
Alessio Lapolla and Aljaž Godec.
\newblock {BetheSF}: {Efficient} computation of the exact tagged-particle
  propagator in single-file systems via the {Bethe} eigenspectrum.
\newblock {\em Comput. Phys. Commun}, page 107569, August 2020.

\bibitem{redner_guide_2007}
Sidney Redner.
\newblock {\em A guide to first-passage processes}.
\newblock Cambridge Univ. Press, Cambridge, digitally printed version (with
  corrections) 2007 edition, 2007.
\newblock OCLC: 830642837.

\bibitem{godec_first_2016}
Aljaž Godec and Ralf Metzler.
\newblock First passage time distribution in heterogeneity controlled kinetics:
  going beyond the mean first passage time.
\newblock {\em Sci Rep}, 6(1):20349, April 2016.

\bibitem{mattos_first_2012}
Thiago~G. Mattos, Carlos Mejía-Monasterio, Ralf Metzler, and Gleb Oshanin.
\newblock First passages in bounded domains: {When} is the mean first passage
  time meaningful?
\newblock {\em Phys. Rev. E}, 86(3):031143, September 2012.

\bibitem{lapolla_httpsgitlabcomsantafe1singlefileabsorbing_2022}
Alessio Lapolla.
\newblock https://gitlab.com/santafe1/{SingleFileAbsorbing}, 2022.

\bibitem{lapolla_manifestations_2019}
Alessio Lapolla and Aljaž Godec.
\newblock Manifestations of {Projection}-{Induced} {Memory}: {General} {Theory}
  and the {Tilted} {Single} {File}.
\newblock {\em Front. Phys.}, 7:182, November 2019.

\bibitem{banerjee_elucidating_2017}
Kinshuk Banerjee, Anatoly~B. Kolomeisky, and Oleg~A. Igoshin.
\newblock Elucidating interplay of speed and accuracy in biological error
  correction.
\newblock {\em Proc Natl Acad Sci USA}, 114(20):5183--5188, May 2017.

\bibitem{abramowitz_milton_and_stegun_irene_a_handbook_1964}
{Abramowitz, Milton and Stegun, Irene A.}
\newblock {\em Handbook of {Mathematical} {Functions} with {Formulas},
  {Graphs}, and {Mathematical} {Tables}}.
\newblock Dover, New York, 1964.

\end{thebibliography}

\end{document}